\documentclass[prb,twocolumn,superscriptaddress,aps]{revtex4-2}

\usepackage{amsmath,amssymb,amsfonts}
\usepackage{graphicx}
\usepackage{bm}
\usepackage[colorlinks=true,linkcolor=blue,citecolor=blue,urlcolor=blue]{hyperref}
\usepackage{xcolor}
\usepackage{amsthm}

\newtheorem{theorem}{Theorem}
\newtheorem{corollary}{Corollary}
\newtheorem{lemma}{Lemma}
\newtheorem{proposition}{Proposition}
\theoremstyle{remark}
\newtheorem{remark}{Remark}

\newcommand{\kB}{k_{\mathrm{B}}}
\newcommand{\kzero}{\kappa_0}
\newcommand{\half}{\tfrac{1}{2}}
\newcommand{\quarter}{\tfrac{1}{4}}

\begin{document}

\title{Topological and kinetic origins of fractional thermal conductance at the topological insulator--superconductor interface}

\author{Laurence A. Jacobs}
\email{laurence.jacobs@uzh.ch}
\affiliation{Centro de Ciencias de la Complejidad, Universidad Nacional Aut\'onoma de M\'exico, CDMx, Mexico}
\affiliation{Center for Molecular Cardiology, University of Zurich, 8952 Schlieren, Switzerland}
\author{Guillermo R. Zemba}
\affiliation{Facultad de Ingenier\'ia y Ciencias Agrarias,  Pontificia Universidad Cat\'olica Argentina,Av. Alicia Moreau de Justo 1500,(C1107AAZ) Buenos Aires, Argentina}
\affiliation{Departamento de F\'{\i}sica Te\'orica de Interacciones Fundamentales y Sistemas Complejos, Laboratorio Tandar,Comisi\'on Nacional de Energ\'{\i}a At\'omica, Av.Libertador 8250,(C1429BNP) Buenos Aires, Argentina}

\date{\today \\[0.5ex] {\small (\protect\text{arXiv}:2401.XXXXX [cond-mat])}}

\begin{abstract}
\noindent A recent experiment by Roy \emph{et al.}\ [Nat.\ Commun.\ \textbf{17}, 2853 (2026)] demonstrated a robust half-integer thermal conductance plateau, $\half\kzero T$, at a bipolar $(\nu,\nu')=(2,-1)$ junction in bilayer graphene, produced not by non-Abelian topology but by full equilibration of co-propagating electron and hole edge modes. We prove a theorem that fixes when the two origins can be told apart: under full thermal equilibration, the two-terminal thermal conductance of a network of chiral edge segments joined at ideal floating contacts is a rational function of the net chiral central charges of the segments alone, quantities pinned by the gravitational anomaly and invariant under arbitrary local boundary kinetics. The corollary is an impossibility statement: whenever two realizations present the same anomaly data to the same network, no thermal-conductance measurement can distinguish them. That the equilibrated $(2,-1)$ value equals the central charge of a chiral Majorana mode is an arithmetic fact about one filling combination, but wherever such a coincidence occurs it is beyond the reach of thermometry, and the separation must come from the charge sector, which is not anomaly-pinned at a superconducting boundary. We develop the topological insulator--superconductor interface as the application: a vortex carries fractional charge $e/4$ from the $\theta=\pi$ magnetoelectric coupling, the boundary hosts a chiral Majorana mode of central charge $\half$, and a laterally adjacent integer quantum Hall (IQH) region supplies the kinetic realization. The Lorenz ratio (anomalous for the isolated Majorana boundary, $L_0/(1+4|\nu_{\mathrm{IQH}}|)$ in the composite device) and the excess shot noise (growing with $B$ from $e/4$ vortex tunneling versus locked to IQH plateaus) together resolve the mechanism, provided the quantum Hall edge does not abut the superconductor: direct contact Andreev-converts the charged modes, leaving the net central charge and $\kappa_{\mathrm{th}}$ exactly unchanged while destroying the charge sector that alone could decide the question. That such a device-level blindness is possible is itself an instance of the theorem.
\end{abstract}

\maketitle

\section{Introduction}

The thermal conductance $\kappa_{\mathrm{th}}$ is a celebrated probe of topological order. For chiral edge theories it is quantized in units of $\kzero T \equiv (\pi^2 \kB^2/3h)T$ with a coefficient $c$ equal to the central charge of the edge conformal field theory, so half-integer values $\half \kzero T$ are conventionally taken as direct evidence of a Majorana edge mode and hence of non-Abelian topological order \cite{KaneFisher1997,Banerjee2018,PaulHeiblum2024,Moore1991}. The most celebrated instances are the $\nu=5/2$ fractional quantum Hall (FQH) state and the Kitaev spin liquid \cite{Matsuda2025}, the latter's half-integer thermal Hall plateau having been reported in $\alpha$-RuCl$_3$ \cite{Kasahara2018,Yokoi2021}.

A recent experiment by Roy \emph{et al.}\ \cite{RoyNatComm2026} has put pressure on this reading. The authors engineer a bipolar $n$--$p$--$n$ heterojunction in bilayer graphene with bulk filling $\nu_{\mathrm{BG}}=2$ and a gated central region at $\nu_{\mathrm{TG}}=-1$, and measure a two-terminal thermal conductance plateau at $\half \kzero T$. The mechanism is not non-Abelian topology: it is full equilibration of co-propagating electron and hole edges. For the $\nu$--$\nu'$--$\nu$ geometry --- two outer regions at filling $\nu$ flanking a gated central region at $\nu'$, with a floating contact at zero net chemical potential serving as the noise thermometer --- the predicted equilibrated value
\begin{equation}
\kappa_{\nu \nu' \nu} \;=\; \frac{|\nu||\nu'|}{|\nu|+2|\nu'|}\,\kzero T \;=\; \half\,\kzero T \qquad (\nu=2,\;\nu'=-1)
\label{eq:RoyFormula}
\end{equation}
is recovered to better than $10\%$, and the same value is reproduced in a second device with $(\nu,\nu')=(-2,1)$. The three indices of $\kappa_{\nu\nu'\nu}$ record the three regions traversed in series; the factor $2$ multiplying $|\nu'|$ in \eqref{eq:RoyFormula} is the multiplicity with which the outer phase is traversed, as we recover in Sec.~\ref{sec:theorem}. The Wiedemann--Franz law $L \equiv K/(T G) = L_0 \equiv \pi^2 \kB^2 / 3 e^2$ holds, in sharp contrast to anyonic heat flow \cite{Banerjee2017,Srivastav2021,Srivastav2022}.

Roy \emph{et al.}\ frame their result as a challenge: \emph{robust fractional thermal conductance can arise from mundane equilibration dynamics, not topology}. In their Discussion they explicitly invite extensions to other platforms, in particular to systems exhibiting charge fractionalization without non-Abelian topology, and they note in passing that half-integer $\kappa_{\mathrm{th}}$ without non-Abelian modes may appear in superconducting contexts \cite{GlodzikSedlmayr2023}. The present paper takes up that invitation --- and sharpens it into a theorem.

\emph{Main result.} The central result of this paper is not a new platform but a general theorem, which we state here in words and prove in Sec.~\ref{sec:main}. Under full thermal equilibration, the two-terminal thermal conductance of \emph{any} network of chiral edge segments joined at ideal floating contacts depends on the microscopic content of each segment only through its net chiral central charge (the coefficient of the boundary gravitational anomaly) and otherwise only through the connectivity of the network. We call this \emph{anomaly pinning}. Its immediate corollary is an impossibility statement: \emph{no thermal-conductance measurement, in any equilibrated geometry, can distinguish two boundaries that share the same anomaly coefficient}, whatever the microscopic origin of that coefficient. A half-integer plateau produced by a genuine chiral Majorana mode and one produced by kinetic equilibration of ordinary integer edges are, at the level of $\kappa_{\mathrm{th}}$, the same measurement. We are careful about what this does and does not assert. It does not explain the coincidence Roy \emph{et al.}\ exhibit: that the equilibrated value of the $(2,-1)$ junction happens to equal $c_M = \half$ is an arithmetic accident of that filling combination, not a consequence of anomaly pinning. What the theorem establishes is that the accident is irreparable by thermometry --- once two realizations present the same anomaly data to the same network, no measurement of $\kappa_{\mathrm{th}}$, however accurate, separates them.

This reframes the entire topology--versus--kinetics question. If $\kappa_{\mathrm{th}}$ cannot separate the mechanisms as a matter of principle, the separation must come from a quantity that is \emph{not} anomaly-pinned. At a superconducting boundary the charge sector is exactly such a quantity: charge conservation is broken by the condensate, so the electrical conductance and the shot noise are free of the anomaly constraint that ties down the heat current. The resolution of the topology--versus--kinetics distinction thus moves, necessarily, into the charge sector --- and the Lorenz ratio and the excess shot noise are, we argue, the \emph{unique} observables that accomplish it.

\emph{Application.} The natural arena in which to exhibit both the theorem and its charge-sector resolution is the topological insulator--superconductor (TI--SC) interface \cite{FuKane2008,QiHughesZhang2008}, and we develop it as the application occupying the second half of the paper.
The chiral Majorana boundary mode invoked below is not automatic: the $s$-wave condensate preserves time-reversal symmetry, and the boundary between two time-reversal-symmetric gapped regions carries counter-propagating Majorana partners with $c-\bar{c}=0$ and hence no net chiral thermal response. Time-reversal symmetry must be broken on the far side of the wall, as at the Fu--Kane superconductor--magnet domain wall. Here the perpendicular field $B$ does double duty: it magnetically gaps the exposed TI surface, converting the boundary of the SC region into a chiral Majorana wall with $c_M=\half$, and it threads the SC region itself with the vortex lattice.
A magnetic vortex piercing the interface carries fractional electric charge $e/4$ per flux quantum, an exact consequence of the $\theta=\pi$ magnetoelectric coupling at the TI surface dressed by the superconducting condensate \cite{HosurGhaemiMongVishwanath2011,QiHughesZhang2008,Jacobs2026}, while the boundary itself hosts a chiral Majorana mode of central charge $c_M=\half$. Two physically distinct mechanisms --- a topological one fixed by the boundary anomaly at $\half\kzero T$, and a kinetic one generated by equilibration with a laterally adjacent integer quantum Hall (IQH) region across a line junction --- act at the \emph{same} physical boundary and, by the theorem, produce thermal responses that no $\kappa_{\mathrm{th}}$ measurement can tell apart. The charge sector separates them: the topological response has anomalous Lorenz ratio and $\sqrt{B}$ shot noise from $e/4$ vortex tunneling, while the kinetic response obeys Wiedemann--Franz and is plateau-locked. This platform is the natural host precisely because the Majorana (neutral, anomaly-pinned) and magnetoelectric (charged, not pinned) sectors are simultaneously present.

The paper is organized as follows. Section~\ref{sec:main} states and proves the anomaly-pinning theorem and its impossibility corollary from the boundary conformal field theory alone, with no reference to any device --- this is the paper's main result. Section~\ref{sec:setup} introduces the TI--SC interface, the $e/4$ charge, and the gravitational-$\theta$ partner. Section~\ref{sec:two_pillars} exhibits the topological and kinetic mechanisms as two instances of the theorem and makes explicit the assumptions behind the device-specific prediction \eqref{eq:kappa_kin}. Section~\ref{sec:diagnostic} develops the Lorenz-ratio and shot-noise diagnostics as the unique charge-sector resolution. Section~\ref{sec:discussion} draws the general lesson, and Sec.~\ref{sec:open} states three open problems whose resolution would complete the theory, the first of them being the specification and mode content of the device itself.

\section{The main result: anomaly pinning of equilibrated thermal transport}
\label{sec:main}

We state and prove the central result of the paper before introducing any device. The only physical input is the universal energy current carried by a chiral edge conformal field theory (CFT); no property of the topological insulator, the superconductor, or the vortex is used. This is deliberate: the theorem is a statement about equilibrated chiral heat transport in general, and the TI--SC interface of Secs.~\ref{sec:setup}--\ref{sec:diagnostic} is one application.

\subsection{Boundary-CFT input}
\label{sec:cft}

A chiral edge branch of central charge $c$ at temperature $T$ carries the universal energy current \cite{CappelliHuertaZemba2002,Kapustin2021}
\begin{equation}
I^E(T) \;=\; \frac{\pi}{12}\, c\, T^2 \qquad (\kB = \hbar = 1),
\label{eq:IE}
\end{equation}
from which the Cappelli--Huerta--Zemba formula for the thermal Hall conductance of a gapped phase with chiral boundary modes follows \cite{CappelliHuertaZemba2002},
\begin{equation}
\kappa_{xy} \;=\; c\, \kzero T, \qquad c \equiv c_R - c_L,
\label{eq:CHZ}
\end{equation}
where $c_R - c_L$ is the net chiral central charge of the edge CFT --- a manifestly physical, frame-independent boundary quantity that captures the gravitational anomaly of the edge without any bulk gravitational-Chern--Simons calculation. Equations \eqref{eq:IE}--\eqref{eq:CHZ} are the sole dynamical input to the theorem below.

\begin{figure*}[t]
\centering
\includegraphics[width=0.92\textwidth]{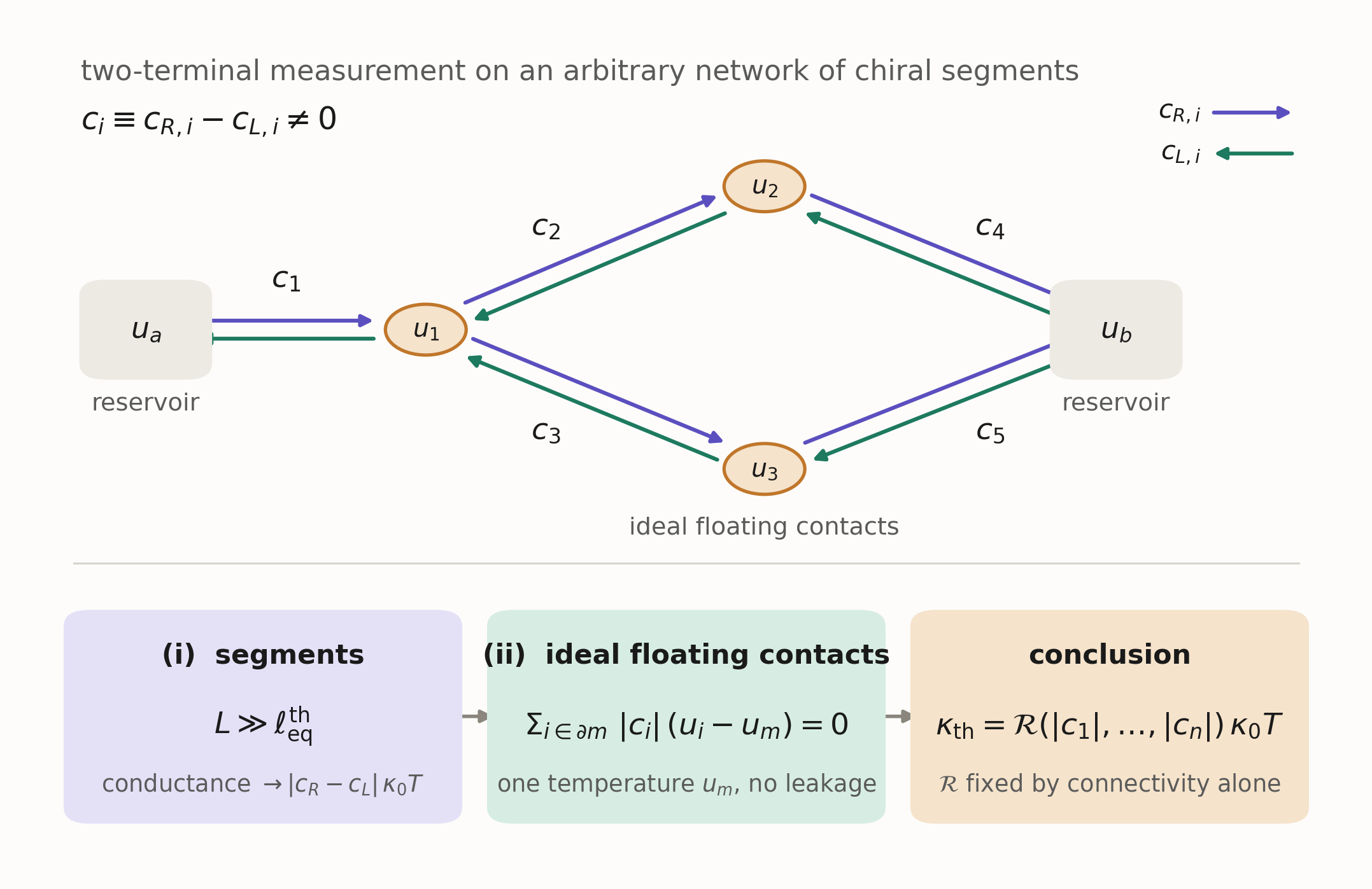}
\caption{The setting of Theorem~\ref{thm:pinning}. Chiral edge segments of net central charge $c_i = c_{R,i}-c_{L,i}$, each long compared with its thermal equilibration length, joined at ideal floating contacts, measured between two reservoirs at $u_a = T_a^2$ and $u_b = T_b^2$. The connectivity is arbitrary; the three-segment chain of Fig.~\ref{fig:chain} is one case. The two hypotheses and the conclusion are the three panels: everything microscopic about a segment enters only through $c_i$, and the rational function $\mathcal{R}$ is fixed by the connectivity alone.}
\label{fig:theorem}
\end{figure*}

\subsection{The theorem}
\label{sec:theorem}

\begin{theorem}[Anomaly pinning of equilibrated thermal transport]
\label{thm:pinning}
Consider a two-terminal thermal transport measurement on a device whose interior consists of (i) chiral edge segments $i=1,\dots,n$, each characterized by downstream and upstream central charges $(c_{R,i},\,c_{L,i})$ with $c_i \equiv c_{R,i} - c_{L,i} \neq 0$, each of length large compared to its thermal equilibration length; and (ii) ideal floating contacts, each of which absorbs all incident energy flux and re-emits every outgoing branch thermally at a \emph{single} contact temperature. Then the two-terminal thermal conductance is
\begin{equation}
\kappa_{\mathrm{th}} \;=\; \mathcal{R}\bigl(|c_1|,\dots,|c_n|\bigr)\,\kzero T,
\label{eq:pinning}
\end{equation}
where $\mathcal{R}$ is a rational function determined solely by the connectivity of the network. In particular, $\kappa_{\mathrm{th}}$ depends on the microscopic content of each segment only through its net chiral central charge $c_i$, which is the coefficient of the boundary gravitational anomaly and is invariant under arbitrary local interactions, disorder, tunneling, and equilibration dynamics on that segment.
\end{theorem}

We prove the theorem, under the stated assumptions, for the three-segment series geometry that all devices in this paper (and in Roy \emph{et al.}) realize, and state the general-network case as a proposition at the end of the section.

\begin{proof}[Proof (physics-level)]
\emph{Step 1: each segment relaxes to its anomaly value.} Within a segment, let $T_R(x)$ and $T_L(x)$ be the local temperatures of the downstream and upstream branches; full equilibration among co-propagating modes ensures each branch carries a single local temperature. By \eqref{eq:IE} the local energy currents are $I^E_{R/L}(x) = \tfrac{\pi}{12}\,c_{R/L}\,T_{R/L}^2(x)$, so the net current $J = \tfrac{\pi}{12}\bigl[c_R T_R^2(x) - c_L T_L^2(x)\bigr]$ is constant in $x$ by energy conservation. Inter-branch equilibration is a local energy exchange proportional to the temperature imbalance; in the squared-temperature variable $v \equiv T_R^2 - T_L^2$ it takes the relaxation form
\begin{equation}
\frac{dv}{dx} \;=\; -\frac{1}{\ell_{\mathrm{eq}}^{\mathrm{th}}}\,\bigl(v - v_\ast\bigr),
\label{eq:relax}
\end{equation}
where $v_\ast$ is the local fixed point at which the two branches carry equal energy current density. Equation \eqref{eq:relax} is solved by $v(x) = v_\ast + [v(0)-v_\ast]\,e^{-x/\ell_{\mathrm{eq}}^{\mathrm{th}}}$, so for a segment of length $L \gg \ell_{\mathrm{eq}}^{\mathrm{th}}$ the boundary imbalance decays exponentially and the interior sits at the fixed point up to corrections $O(e^{-L/\ell_{\mathrm{eq}}^{\mathrm{th}}})$. Imposing the reservoir boundary conditions $T(0)=T_a$, $T(L)=T_b$ and eliminating the equilibrated interior profile gives
\begin{equation}
J \;=\; \frac{\pi}{12}\,|c_R - c_L|\,\bigl(T_a^2 - T_b^2\bigr)\;\bigl[1 + O(e^{-L/\ell_{\mathrm{eq}}^{\mathrm{th}}})\bigr],
\label{eq:segment}
\end{equation}
i.e.\ a linear-response segment conductance $|c_R - c_L|\,\kzero T$. The coefficient cannot be renormalized by any local boundary kinetics: $c_R - c_L$ is the gravitational-anomaly coefficient of the segment's edge theory \cite{ryu2012electromagnetic,CappelliHuertaZemba2002,Kapustin2021}, fixed by the adjacent bulk through the bulk--boundary correspondence, and \eqref{eq:relax} shows equilibration \emph{flows the segment to it}. The ballistic ($L \ll \ell_{\mathrm{eq}}^{\mathrm{th}}$) limit instead retains the non-universal sum $(c_R + c_L)\,\kzero T$, and finite $L/\ell_{\mathrm{eq}}^{\mathrm{th}}$ interpolates via the same exponential.

The rate equation \eqref{eq:relax} is written here in relaxation-time form, with $\ell_{\mathrm{eq}}^{\mathrm{th}}$ an input. A microscopic derivation of $\ell_{\mathrm{eq}}^{\mathrm{th}}$, and of the full crossover function at finite $L/\ell_{\mathrm{eq}}^{\mathrm{th}}$ rather than only its two endpoints, would upgrade \eqref{eq:segment} from an asymptotic statement to a complete one and is exactly the line-junction kinetic problem of Refs.~\cite{Protopopov2017,Nosiglia2018}; carrying it out for the composite boundary of Sec.~\ref{sec:kinetic} would additionally supply the length-dependent Lorenz profile (Sec.~\ref{sec:open}).

\emph{Step 2: networks give rational combinations.} We first record the property of an ideal floating contact that the network algebra requires.

\begin{lemma}[Ideal floating contact]
\label{lem:contact}
A floating contact that (a) absorbs all incident energy flux, (b) exchanges no energy with any reservoir, and (c) re-emits into every outgoing branch at a single temperature $T_m$, sits at the temperature fixed by
\begin{equation}
\sum_{i \,\in\, \partial m} |c_i|\,\bigl(u_i - u_m\bigr) \;=\; 0,
\qquad u \equiv T^2,
\label{eq:contact}
\end{equation}
the sum running over the branches incident on $m$. Hence $u_m$ is the $|c_i|$-weighted mean of its neighbours' squared temperatures.
\end{lemma}

\begin{proof}[Proof of Lemma]
Each incident branch delivers energy current $\tfrac{\pi}{12}|c_i|\,u_i$ and each emitted branch removes $\tfrac{\pi}{12}|c_i|\,u_m$, by \eqref{eq:segment} and hypothesis (c). Condition (a)--(b) is steady-state energy balance at the contact, $\sum_i \tfrac{\pi}{12}|c_i|(u_i - u_m) = 0$, which is \eqref{eq:contact}. \end{proof}

By \eqref{eq:segment} each segment carries a heat current linear in the $u$'s with coefficient $\tfrac{\pi}{12}|c_i|$, so in the linear-response limit (all reservoir temperatures close to a common $T$, $u_a - u_b \to 2T\,\Delta T$) the network is a \emph{linear} resistor network in the variables $u$, and Lemma~\ref{lem:contact} supplies one linear constraint per floating contact. Eliminating the internal $\{u_m\}$ from this linear system expresses the two-terminal current as a ratio of determinants in the $|c_i|$ --- a rational function whose structure is fixed by the connectivity alone.

We carry this out explicitly for the geometry realized by the devices in this paper and in Roy \emph{et al.}: three segments in series, of net charges $c_1, c_2, c_3$, joined by two ideal floating contacts $M_1$ and $M_2$, with reservoirs $a$ and $b$ at the two ends. Lemma~\ref{lem:contact} at $M_1$ and at $M_2$ states that the current is the same across each segment,
\begin{equation}
c_1\,(u_a - u_1) \;=\; c_2\,(u_1 - u_2) \;=\; c_3\,(u_2 - u_b) \;=\; \tfrac{12}{\pi}\,J,
\label{eq:Mbalance}
\end{equation}
and eliminating the two interior temperatures gives the series law $\mathcal{R}^{-1} = c_1^{-1} + c_2^{-1} + c_3^{-1}$. This is the content of the label $\kappa_{\nu\nu'\nu}$ in \eqref{eq:RoyFormula}: the three indices are the three regions traversed. When the two outer segments are identical, $c_1 = c_3 \equiv c_{\parallel}$ and $c_2 \equiv c_{\perp}$, the coefficient is
\begin{equation}
\mathcal{R}(c_{\parallel}, c_{\perp}) \;=\; \Bigl(\frac{2}{c_{\parallel}} + \frac{1}{c_{\perp}}\Bigr)^{-1} \;=\; \frac{c_{\parallel}\,c_{\perp}}{c_{\parallel} + 2\,c_{\perp}},
\label{eq:R3arm}
\end{equation}
which is Eq.~\eqref{eq:RoyFormula} at $(c_{\parallel},c_{\perp}) = (|\nu|,|\nu'|)$. The function \eqref{eq:R3arm} is \emph{not} symmetric in its two arguments: the phase traversed twice enters as $2/c_{\parallel}$ and the phase traversed once as $1/c_{\perp}$. The factor $2$ in the Roy formula is therefore a multiplicity of traversal, not a branching of the circuit, and the asymmetry is a physical statement about which phase forms the central region and which the two flanking ones --- same central charges, different placement in the network, different $\kappa_{\mathrm{th}}$.\footnote{An earlier version of this section derived \eqref{eq:R3arm} from a one-in/two-out star: a single floating contact fed by one branch of charge $c_{\parallel}$ from reservoir $a$ and two branches of charge $c_{\perp}$ to reservoir $b$. That circuit does not reproduce \eqref{eq:RoyFormula}; eliminating $u_M$ there gives $2c_{\parallel}c_{\perp}/(c_{\parallel}+2c_{\perp})$, twice \eqref{eq:R3arm}. The series chain above is the correct realization, and it also fixes the sense of the asymmetry: the charge entering with weight $2$ is the one \emph{traversed} twice, i.e.\ the outer regions, not a doubled arm. We record the correction rather than silently deleting it, since the two readings assign the Majorana and IQH sectors to opposite roles in Sec.~\ref{sec:kinetic}.}
\end{proof}

Lemma~\ref{lem:contact} was applied above to \emph{simple} segments, each of definite net charge $|c_i|$. When a floating contact is fed by a \emph{composite} segment carrying two branches of different character --- a charge-blind Majorana branch of charge $\half$ and a charge-carrying IQH branch of charge $|\nu_{\mathrm{IQH}}|$, as at the TI--SC--IQH boundary of Sec.~\ref{sec:kinetic} --- hypothesis (c) of the lemma is a nontrivial physical condition rather than a definition, and the value of $\kappa_{\mathrm{th}}$ forks on whether it holds. We state the fork sharply, because closing it is where the theorem can be strengthened. We say at once that the fork presupposes something prior: that a composite segment with branches $\{\half, |\nu_{\mathrm{IQH}}|\}$ exists as described, i.e.\ that the charged branch is still charged where it meets the Majorana branch. Section~\ref{sec:kinetic} shows that this is a statement about the device geometry and not a given, and Sec.~\ref{sec:open} states it as the first open problem.

\begin{proposition}[Composite contact: conditional value]
\label{prop:composite}
Let a floating contact be fed by a composite segment of branches $\{\half,\,|\nu_{\mathrm{IQH}}|\}$.
\begin{enumerate}
\item[(i)] \emph{Thermalized.} If the two branches equilibrate to a single contact temperature (hypothesis (c) holds), the contact presents one temperature to the network, Lemma~\ref{lem:contact} applies verbatim with the branches entering as segments of charge $\half$ and $|\nu_{\mathrm{IQH}}|$ in the positions the device assigns them, and the series value \eqref{eq:R3arm} follows at $(c_\parallel, c_\perp) = (\half, |\nu_{\mathrm{IQH}}|)$ (the Majorana boundary flanking, the IQH region central), giving $\kappa_{\mathrm{th}} = |\nu_{\mathrm{IQH}}|/(1+4|\nu_{\mathrm{IQH}}|)\,\kzero T$.
\item[(ii)] \emph{Unthermalized.} If the branches maintain distinct temperatures $T_M^{(M)} \neq T_M^{(\mathrm{IQH})}$, the contact imposes \emph{two} conservation laws \eqref{eq:contact}, one per branch, and the composite no longer reduces to a single segment. The two-terminal value is then a different rational function of $\{\half,|\nu_{\mathrm{IQH}}|\}$, whose form depends on the return-path connectivity the two independent branches realize. Two cases make the point. If each branch closes through its own counterpart, the device is two independent chains conducting in parallel and their coefficients add,
\begin{equation}
\mathcal{R} \;=\; \mathcal{R}(\half,\half) + \mathcal{R}(|\nu_{\mathrm{IQH}}|,|\nu_{\mathrm{IQH}}|) \;=\; \tfrac{1}{6} + \tfrac{1}{3}|\nu_{\mathrm{IQH}}|,
\label{eq:caseII_a}
\end{equation}
each term being a uniform three-segment chain. If instead the branches cross, the two chains are the two arrangements of Sec.~\ref{sec:kinetic} in parallel,
\begin{equation}
\begin{split}
\mathcal{R} &=\; \frac{|\nu_{\mathrm{IQH}}|}{1+4|\nu_{\mathrm{IQH}}|} + \frac{|\nu_{\mathrm{IQH}}|}{2(|\nu_{\mathrm{IQH}}|+1)} \\
&=\; \frac{3|\nu_{\mathrm{IQH}}|\bigl(2|\nu_{\mathrm{IQH}}|+1\bigr)}{2\bigl(|\nu_{\mathrm{IQH}}|+1\bigr)\bigl(4|\nu_{\mathrm{IQH}}|+1\bigr)}.
\end{split}
\label{eq:caseII_b}
\end{equation}
Neither value is determined by the anomaly data alone. Selecting the physical connectivity is part of the same kinetic problem that decides between (i) and (ii).
\end{enumerate}
In case (i) $\kappa_{\mathrm{th}}$ is anomaly-pinned in the strict sense of the theorem. In case (ii) the theorem's single-temperature hypothesis fails at the composite contact and $\kappa_{\mathrm{th}}$ acquires connectivity dependence beyond the anomaly data; the two cases are physically distinct and give different plateau values.
\end{proposition}

Proposition~\ref{prop:composite} is the precise sense in which the theorem is presently conditional, and it marks the single most valuable place to strengthen the paper. Proving \emph{which} alternative the physical Andreev/ohmic contact realizes --- i.e.\ proving or refuting hypothesis (c) for the Majorana--IQH contact --- would (1) determine whether the TI--SC--IQH platform falls under the anomaly-pinning theorem at all (case (i)) or instead exhibits connectivity dependence beyond the anomaly data (case (ii)), (2) fix the channel assignment in \eqref{eq:kappa_kin} from first principles, and (3) deliver the length-dependent Lorenz crossover as a byproduct, since the same kinetic equations that settle (c) also govern partial equilibration. This is a self-contained line-junction problem of the Gefen--Mirlin type, and it is the central open problem (Sec.~\ref{sec:open}). The stakes are real but bounded: the \emph{universal} theorem of Sec.~\ref{sec:theorem} stands regardless, since it concerns simple segments for which hypothesis (c) is a definition; what case (ii) would revise is only the theorem's \emph{application} to this particular composite device --- and even a case-(ii) outcome is a publishable sharpening, not a refutation, since it would identify the composite boundary as a concrete setting where anomaly pinning is evaded by contact physics.

\begin{proposition}[General networks]
\label{prop:general}
For a network of the type in Theorem~\ref{thm:pinning} with arbitrary connectivity, the elimination of Step 2 applies unchanged: the conservation laws at the floating contacts form a linear system in the variables $u_m = T_m^2$ with coefficients $|c_i|$, and Cramer's rule expresses $\kappa_{\mathrm{th}}$ as a ratio of determinants in the $\{|c_i|\}$, with structure fixed by the connectivity. We state the general case as a proposition because the three-arm result \eqref{eq:R3arm} is the only one used in this paper; the determinantal form and its consequences for more elaborate multi-contact geometries are left to future work.
\end{proposition}

\emph{Relation to existing practice.} The heat-balance calculus of Steps 1--2 is the one routinely used to interpret two-terminal thermal-conductance experiments on quantum Hall edges \cite{KaneFisher1997,Srivastav2019,Srivastav2022,Protopopov2017,Nosiglia2018}; Theorem~\ref{thm:pinning} does not introduce it, it states what it implies. The content added here is that the calculus admits \emph{only} the anomaly coefficients and the connectivity as inputs, and hence the negative statement that follows.

\begin{figure*}[t]
\centering
\includegraphics[width=0.86\textwidth]{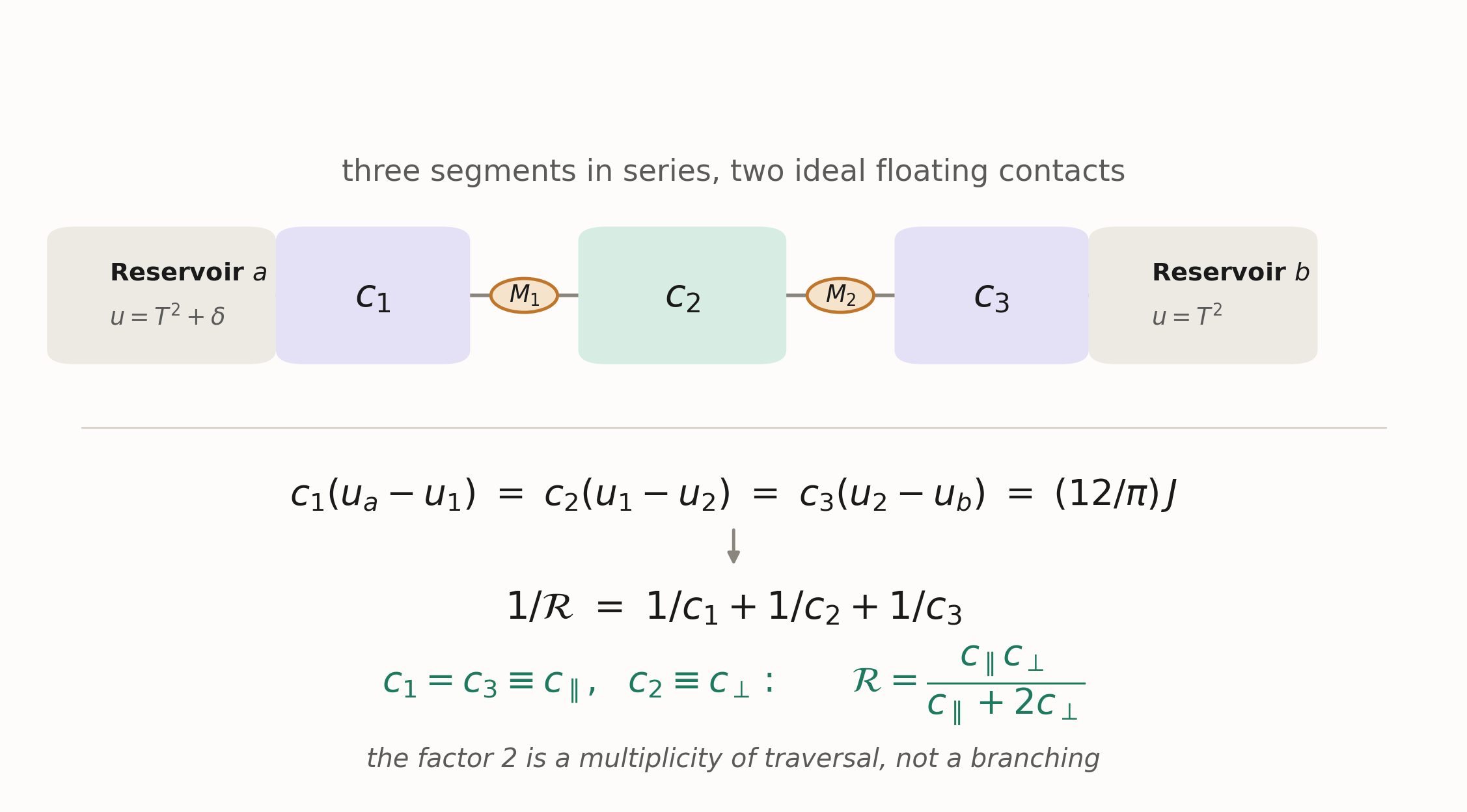}
\caption{The three-segment series geometry realized by the device of Roy \emph{et al.}\ and by the platform of Sec.~\ref{sec:kinetic}. Lemma~\ref{lem:contact} at each floating contact makes the heat current common to the three segments, and eliminating the interior temperatures gives the series law. When the two outer segments are identical the coefficient is \eqref{eq:R3arm}: the factor $2$ in the Roy formula \eqref{eq:RoyFormula} is the multiplicity with which the outer phase is traversed, not a branching of the circuit.}
\label{fig:chain}
\end{figure*}

\begin{corollary}[Impossibility]
\label{cor:impossibility}
No thermal-conductance measurement, in any fully equilibrated two-terminal geometry, can distinguish a boundary whose net central charge $\half$ arises from a single chiral Majorana mode from any other microscopic realization of net chirality $\half$. More generally, two devices with the same network connectivity and the same set $\{|c_i|\}$ are indistinguishable by $\kappa_{\mathrm{th}}$, whatever the microscopic origin of each $c_i$.
\end{corollary}

\begin{remark}[What is \emph{not} pinned]
\label{rem:notpinned}
The electrical conductance is pinned by the $U(1)$ charge anomaly only where charge is conserved. At a superconducting boundary, charge conservation is broken by the condensate, so $G$ is \emph{not} anomaly-pinned; the Lorenz ratio $L = \kappa_{\mathrm{th}}/(TG)$ therefore becomes mechanism-sensitive precisely where $\kappa_{\mathrm{th}}$ does not. Likewise, shot-noise Fano factors are properties of the equilibration fixed point and of the tunneling quasiparticles \cite{MannaDasGoldsteinGefen2024a,MannaDasGoldsteinGefen2025}, not of the anomaly. The diagnostics of Sec.~\ref{sec:diagnostic} are consequently not a convenience but the \emph{only possible} resolution of the topology--versus--kinetics question. This is the structural point of the paper, and it is why the second half moves entirely into the charge sector.
\end{remark}

\begin{remark}[Assumptions and their failure modes]
\label{rem:assumptions}
Theorem~\ref{thm:pinning} assumes ideal floating contacts (complete thermalization at a single temperature, no phonon leakage) and full thermal equilibration in every segment. Each assumption fails in instructive ways: partial equilibration replaces \eqref{eq:segment} by a crossover function of $L/\ell_{\mathrm{eq}}^{\mathrm{th}}$ interpolating between the ballistic sum $(c_R+c_L)$ and the anomaly value $|c_R-c_L|$ \cite{Protopopov2017,Nosiglia2018}; edge--phonon coupling modifies the approach to quantization, and the very observability of the anomaly-pinned value can hinge on it \cite{VinklerAvivRosch2018,YeHalaszSavaryBalents2018}. The single-contact-temperature clause is a genuine physical hypothesis, not a formality: for a contact fed by a \emph{composite} segment carrying both a charge-blind Majorana branch and a charge-carrying IQH branch --- precisely the situation at the TI--SC--IQH line junction of Sec.~\ref{sec:kinetic} --- it is not automatic, and establishing it is one of the open problems of Sec.~\ref{sec:open}. Prior to it stands the question of whether that composite segment exists with the stated mode content at all, which is a matter of device geometry (Sec.~\ref{sec:kinetic}). Because the Majorana branch carries heat but no charge, the thermal and charge equilibration crossovers decouple, so at partial equilibration the measured Lorenz ratio flows with junction length between its two limits (Sec.~\ref{sec:diagnostic}) --- a device-level refinement we return to in the Discussion.
\end{remark}

\begin{remark}[Edge networks and junction conditions]
\label{rem:Mintchev}
In the literature on quantum wire networks and edge junctions, vertex conditions for charge and particle currents are formulated through scattering matrices satisfying conservation laws and unitarity at each node, and the formalism extends to energy and thermal currents in similar terms, with a common local temperature assumed at the junction \cite{minchevsorba,CAUDRELIER20131788,KHACHATRYAN2010444,PhysRevB.79.085122}. Lemma~\ref{lem:contact} is the equilibrated-contact counterpart of such vertex conditions.
\end{remark}

\section{Setup}
\label{sec:setup}

\subsection{Geometry}

We consider a planar interface between a three-dimensional time-reversal-invariant topological insulator and an $s$-wave superconductor, in the geometry of Refs.~\cite{FuKane2008,Jacobs2026}. A perpendicular magnetic field $B$ threads the interface, generating a triangular vortex lattice of areal density $n_v = B/\Phi_0$, where $\Phi_0 = h/2e$ is the superconducting flux quantum. We take the superconductor to be type-II throughout, and note that this is a genuine requirement rather than a convenience: the Abrikosov lattice of singly quantized vortices assumed here --- and with it the per-flux-quantum charge \eqref{eq:e_over_4}, the Fu--Kane zero-mode array, and the counting $N_v = BA/\Phi_0$ used in Sec.~\ref{sec:diagnostic} --- has no type-I counterpart, since below the Ginzburg--Landau threshold flux penetrates as macroscopic normal domains of the intermediate state rather than as quantized vortices. The TI bulk is described by the effective action
\begin{equation}
S_{\theta} \;=\; \frac{\theta}{32\pi^2}\!\int\! d^4x\, \epsilon^{\mu\nu\rho\sigma} F_{\mu\nu} F_{\rho\sigma},\qquad \theta = \pi,
\label{eq:axion_em}
\end{equation}
which is the electromagnetic axion term \cite{QiHughesZhang2008}. The boundary of the TI--SC region is interfaced with an adjacent IQH region at integer filling $\nu_{\mathrm{IQH}}$. We adopt the $\nu$--$\nu'$--$\nu$ series geometry of Roy \emph{et al.}\ throughout, with a floating contact serving as the noise thermometer, so that the effective chemical potential of the central contact is held at zero by simultaneous injection of equal electron and hole currents from two identical reservoirs, eliminating shot noise from current partitioning and isolating the Johnson--Nyquist signal that encodes the temperature rise $T_M$.

\begin{figure*}[t]
\centering
\includegraphics[width=0.88\textwidth]{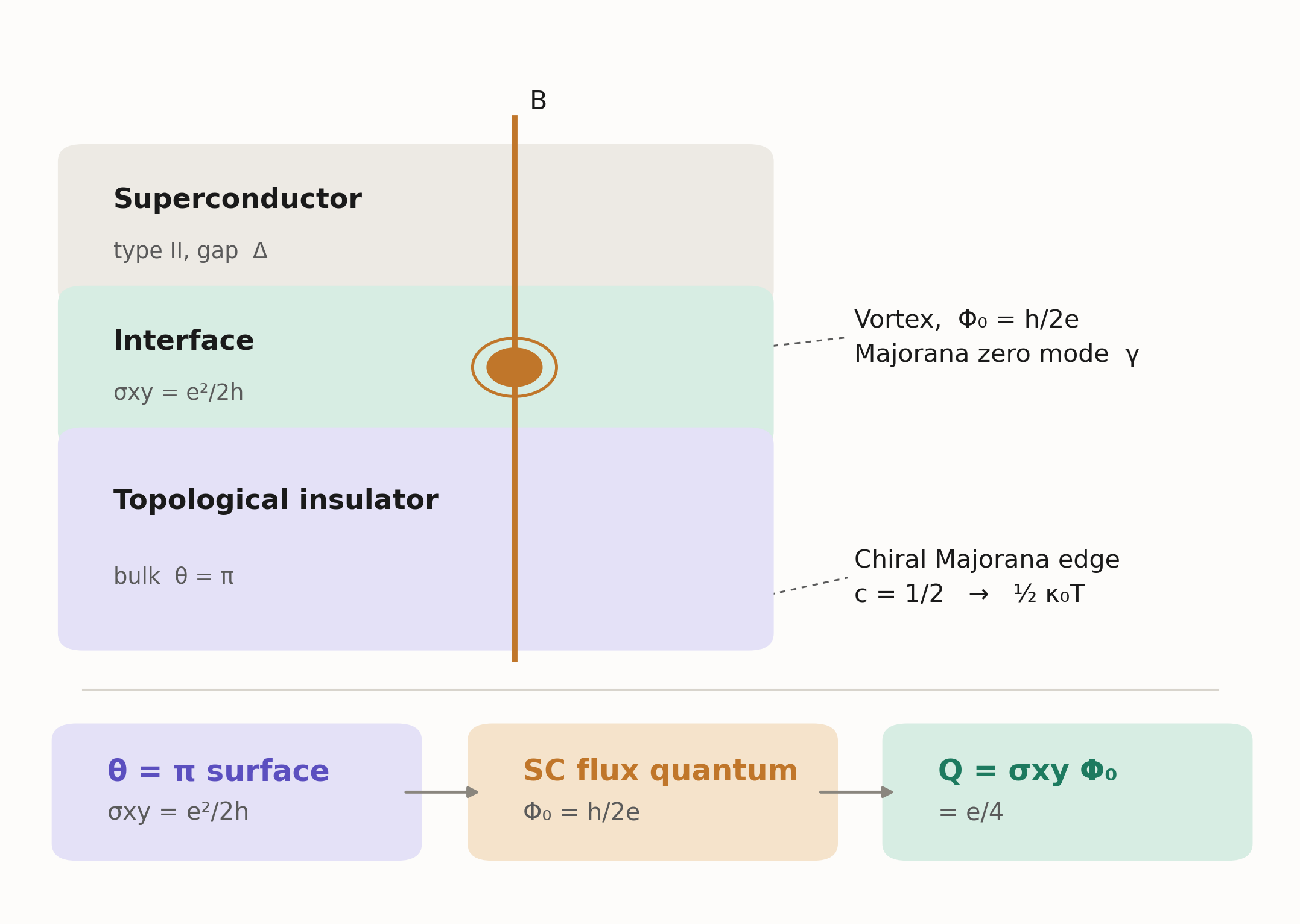}
\caption{The TI--SC interface. A perpendicular field threads the interface and generates an Abrikosov lattice of singly quantized vortices; each carries a Fu--Kane Majorana zero mode at its core and, from the $\theta=\pi$ magnetoelectric coupling, the polarization charge \eqref{eq:e_over_4}. The chain below records the origin of the $e/4$: the half-quantized surface Hall conductance of the $\theta=\pi$ surface, multiplied by the \emph{superconducting} flux quantum $\Phi_0 = h/2e$. The boundary of the region carries a chiral Majorana mode of central charge $c_M = \half$, which is a property of the boundary and not of any individual vortex.}
\label{fig:tisc}
\end{figure*}

\subsection{Recall: fractional vortex charge \texorpdfstring{$e/4$}{e/4}}

A magnetic flux quantum $\Phi_0$ piercing the TI--SC interface induces a localized electromagnetic polarization charge \cite{HosurGhaemiMongVishwanath2011,QiHughesZhang2008,Jacobs2026}
\begin{equation}
Q_{\mathrm{vortex}} \;=\; \frac{e\,\theta}{4\pi} \;=\; \frac{e}{4} \qquad (\theta = \pi),
\label{eq:e_over_4}
\end{equation}
exact at the level of the topological action and independent of microscopic details of the TI gap, the SC pairing strength, or the precise interface profile, provided $\theta$ remains pinned at $\pi$ by time-reversal symmetry of the bulk TI. The same vortex hosts a Fu--Kane Majorana zero mode at its core \cite{FuKane2008}. The two features, the $e/4$ polarization charge and the Majorana zero mode, are independent in the sense that the former survives gapping of the latter (e.g.\ by intervortex Majorana hybridization in a dense lattice), while the latter survives weak deformations that locally shift the polarization. They are coexisting topological features of the same defect.

\subsection{Gravitational-\texorpdfstring{$\theta$}{theta} partner}

The Ryu--Moore--Ludwig program \cite{ryu2012electromagnetic} establishes that for topological insulators and superconductors in three spatial dimensions, the electromagnetic axion term \eqref{eq:axion_em} is accompanied --- in symmetry classes for which it is well-defined --- by a gravitational topological term
\begin{equation}
S_{\theta_g} \;=\; \frac{\theta_g}{1536\pi^2}\!\int\! d^4x\, \epsilon^{\mu\nu\rho\sigma} R^{\alpha}{}_{\beta\mu\nu} R^{\beta}{}_{\alpha\rho\sigma},
\label{eq:axion_grav}
\end{equation}
where $R^{\alpha}{}_{\beta\mu\nu}$ is the Riemann curvature tensor and $\theta_g$ is quantized modulo $2\pi$ in the same way $\theta$ is. For the time-reversal-invariant TI--SC heterostructure considered here, $\theta_g = \pi$ at the topologically nontrivial side of the interface. The coupling \eqref{eq:axion_grav} encodes, via the standard correspondence between gravitational background and thermal transport (Luttinger \cite{Luttinger1964}), a thermal partner to the magnetoelectric response. As Stone has emphasized \cite{StoneGravitational}, however, extracting the laboratory-frame thermal Hall response directly from \eqref{eq:axion_grav} is subtle: a uniform thermal gradient corresponds in the gravitational analog to a uniform Newtonian potential, which generates no transverse energy current; only gravitational field gradients (tidal forces) do (see Appendix~\ref{app:stone}). Throughout this paper, and already in the theorem of Sec.~\ref{sec:main}, we therefore regard \eqref{eq:axion_grav} as an organizing principle and obtain all quantitative thermal statements from the boundary CFT, whose single dynamical input is the universal chiral energy current \eqref{eq:IE} and the Cappelli--Huerta--Zemba formula \eqref{eq:CHZ} it implies.

\section{Two realizations of the pinned value at the TI--SC interface}
\label{sec:two_pillars}

By Theorem~\ref{thm:pinning}, any equilibrated realization of a given anomaly coefficient produces the same $\kappa_{\mathrm{th}}$. We now exhibit two physically distinct realizations at the TI--SC interface, one topological and one kinetic, and confirm that they are instances of the single pinned value, so that (by Corollary~\ref{cor:impossibility}) $\kappa_{\mathrm{th}}$ cannot separate them. The separation is deferred entirely to the charge sector of Sec.~\ref{sec:diagnostic}.

\subsection{Topological realization: the anomaly of the TI--SC boundary}
\label{sec:top_mech}

The boundary of the TI--SC heterostructure supports a single chiral Majorana edge mode \cite{FuKane2008,QiHughesZhang2008}, with central charge
\begin{equation}
c_M \;=\; \half,
\label{eq:cM_main}
\end{equation}
exactly half that of a chiral Dirac fermion. Substituting \eqref{eq:cM_main} into the Cappelli--Huerta--Zemba formula \eqref{eq:CHZ} gives the thermal Hall conductance of the TI--SC interface,
\begin{equation}
\kappa_{xy}^{\mathrm{top}} \;=\; \half\, \kzero T.
\label{eq:kappa_top}
\end{equation}
Equation \eqref{eq:kappa_top} is written as a Hall response, but for a chiral boundary the two-terminal conductance carries the same coefficient, since the universal energy current \eqref{eq:IE} is the whole of the transport; we therefore compare it directly with the two-terminal values of Sec.~\ref{sec:kinetic}. It is a property of the \emph{entire chiral boundary} of the TI--SC region, not of any individual vortex; it is set by the topological character of the bulk through the bulk--boundary correspondence \cite{Kapustin2021} and survives any local deformation that preserves the gapped TI--SC bulk and the chiral Majorana character of the boundary. It is the thermal counterpart of \eqref{eq:e_over_4} in a precise but limited sense: both are parameter-free consequences of the topological character of the interface; but the $e/4$ is a localized integrated charge, whereas \eqref{eq:kappa_top} is a boundary anomaly.\footnote{An earlier draft of this work framed the topological contribution as $\quarter\kzero T$ \emph{per flux quantum}, by direct analogy with the per-vortex charge \eqref{eq:e_over_4}. That framing over-stretches the electromagnetic analogy; the correct statement is \eqref{eq:kappa_top}, integrated over the boundary.}

The role of vortices is through the bulk spectrum. At finite vortex density $n_v = B/\Phi_0$, the Fu--Kane Majorana zero modes localized at each vortex core hybridize through intervortex tunneling, forming dispersive Majorana bands \cite{MurrayVafek2015}. These bands carry a Chern number $\mathcal{C}_M$ determined by the sign structure of the intervortex hopping amplitudes; when $\mathcal{C}_M \neq 0$ they contribute an additional chiral central charge $\mathcal{C}_M/2$ to the boundary,
\begin{equation}
c_{\mathrm{total}} \;=\; \half + \half\,\mathcal{C}_M(B,\Delta_{\mathrm{SC}},\dots).
\label{eq:c_total}
\end{equation}
For chiral $p$-wave SCs Murray and Vafek find $\mathcal{C}_M \in \{0,\pm 1\}$ in different parameter regimes \cite{MurrayVafek2015}; for the proximity-induced Fu--Kane SC the analysis has not been carried out explicitly, but $\mathcal{C}_M \in \{0,\pm 1\}$ is the natural expectation by the same symmetry-class analysis. In the regime $\mathcal{C}_M = 0$ --- the parameter-free regime assumed below unless stated otherwise --- Eq.~\eqref{eq:kappa_top} stands as written; $\mathcal{C}_M = \pm 1$ shifts the boundary value to $\kzero T$ or zero, and transitions between these regimes as a function of $B$ would appear as steps in the measured thermal response.

The assumption $\mathcal{C}_M = 0$ is not free, and we record what it costs. Intervortex hybridization of the Majorana zero modes is controlled by the ratio of the vortex spacing to the superconducting coherence length: for a triangular lattice, $d = (2\Phi_0/\sqrt{3}B)^{1/2} \simeq 49\,\mathrm{nm}\,(B/1\,\mathrm{T})^{-1/2}$, while the hybridization amplitude falls as $e^{-d/\xi}$. The regime in which the Majorana bands remain trivial and \eqref{eq:kappa_top} holds parameter-free is therefore $d \gg \xi$, which favors a short-coherence-length superconductor at moderate field: for representative values, $\xi \simeq 8$\,nm gives $e^{-d/\xi} \sim 10^{-3}$ at $1$\,T, whereas $\xi \simeq 40$\,nm gives $e^{-d/\xi} \sim 0.3$ at the same field, and the zero modes are strongly hybridized throughout the accessible range. This constraint pulls against the $B^{1/2}$ growth of the topological noise signal \eqref{eq:S_top}, which favors large $N_v$ and hence large $B$: increasing $B$ strengthens the noise diagnostic while driving $d$ toward $\xi$ and thus toward the $\mathcal{C}_M \neq 0$ regimes in which the headline value \eqref{eq:kappa_top} no longer holds. The two requirements define a window rather than a limit, and locating it quantitatively for a given heterostructure requires the microscopic calculation of $\mathcal{C}_M$ stated in Sec.~\ref{sec:open}.

\subsection{Kinetic realization: a line junction with an adjacent quantum Hall region}
\label{sec:kinetic}

\emph{What the device is made of.} The second realization requires a device, and the analysis turns on a point that must be settled before any central charge can be assigned: what the quantum-Hall region consists of, and whether its charged edge modes survive proximity to the superconductor. Three arrangements are physically distinct and they do not give the same answer (Fig.~\ref{fig:geometries}). We state all three, adopt the first, and record that the third --- far from being a nuisance case to be avoided --- is an instructive illustration of Corollary~\ref{cor:impossibility}.

\emph{(A) Line junction with a spacer}, adopted here. The TI--SC region and a laterally adjacent two-dimensional electron gas or graphene layer at integer filling $\nu_{\mathrm{IQH}}$, separated by a normal or insulating strip of width $w$. The quantum-Hall edge never touches the superconductor. Its $|\nu_{\mathrm{IQH}}|$ charged modes therefore remain charged, the region contributes electrical conductance $G = |\nu_{\mathrm{IQH}}|e^2/h$, and the two edges are coupled across the strip by single-particle tunneling and Coulomb interaction. This is a genuine line junction, and it is the object treated by the kinetic theory of Refs.~\cite{Protopopov2017,Nosiglia2018} that we invoke throughout. It is also the arrangement in which the spacer width $w$ becomes a design parameter: wide enough that the charged modes are not converted, narrow enough that the two edges equilibrate over the junction length.

\emph{(B) All on the topological-insulator surface.} If the ``quantum Hall region'' is the same TI surface gated into Landau levels, the assignment of central charges changes and must be made in terms of domain walls rather than regions --- a single magnetically gapped Dirac cone is anomalous by itself and carries $\sigma_{xy} = (n+\half)e^2/h$, so ``a region at filling $\nu$'' is not an object to which a central charge can be attached. A wall between two magnetic regions with $n_1 \neq n_2$ carries $|n_1-n_2|$ chiral Dirac modes and hence integer $c$; a wall between an SC-gapped region and a magnetic region carries $|2n+1|$ chiral Majorana modes and hence $c = n+\half$. The Fu--Kane value $c_M = \half$ is the $n=0$ member of this family, not the only one \cite{FuKane2008,QiHughesZhang2008}. Writing $m \equiv 2n+1$ and $\Delta n$ for the jump across the magnetic wall, the series function \eqref{eq:R3arm} gives
\begin{equation}
\frac{\kappa_{\mathrm{th}}}{\kzero T} \;=\; \frac{m\,|\Delta n|}{m + 4|\Delta n|},
\label{eq:kappa_surface}
\end{equation}
which reduces to \eqref{eq:kappa_kin} below at $m=1$, $|\Delta n| = |\nu_{\mathrm{IQH}}|$, and is otherwise a one-parameter family of plateau values. Geometry (B) is the more economical realization (one surface, no second material), and the extra parameter $m$ yields extra predictions. We do not develop it here because the mode content of the magnetic domain walls, their equilibration, and the contacts between them would have to be established with the same care we are about to demand of (A); we leave it as a variant worth its own treatment.

\emph{(C) Direct contact.} If the quantum-Hall edge abuts the superconductor, it is Andreev-converted along the interface: the charged chiral modes hybridize with their hole partners into chiral Andreev edge states, as studied and observed in graphene quantum-Hall--superconductor devices. This case deserves attention rather than avoidance, both because it is the closest to what has actually been fabricated and because of what the theorem says about it. Andreev conversion does not remove modes: $|\nu_{\mathrm{IQH}}|$ chiral Dirac fermions become $2|\nu_{\mathrm{IQH}}|$ chiral Majoranas, so the net chiral central charge of the segment is \emph{unchanged}, $c = |\nu_{\mathrm{IQH}}|$. By Theorem~\ref{thm:pinning} the thermal conductance is therefore exactly the same as in (A). What the conversion destroys is the charge sector: charge is no longer conserved along the boundary, $G$ is not $|\nu_{\mathrm{IQH}}|e^2/h$, and the Wiedemann--Franz comparison loses its reference value. Geometry (C) is thus a device in which $\kappa_{\mathrm{th}}$ is blind by Corollary~\ref{cor:impossibility} \emph{and} the diagnostic that was to rescue it has been switched off by the contact itself. It is a clean illustration of the thesis of this paper --- and, for the same reason, the one geometry in which the proposed measurement cannot be performed. The spacer in (A) is not a convenience. It is what preserves the only sector that can answer the question.

\begin{figure*}[t]
\centering
\includegraphics[width=0.92\textwidth]{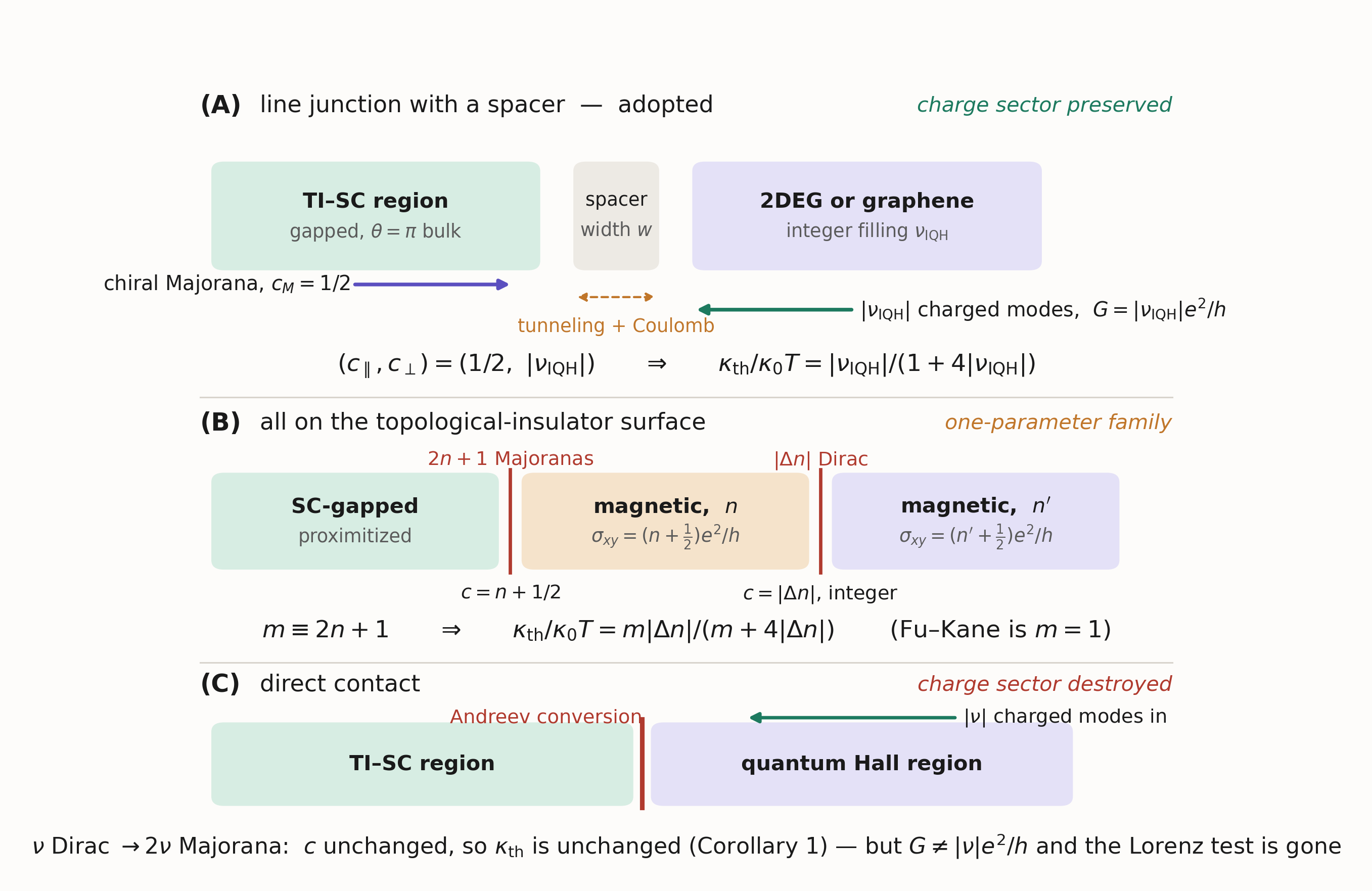}
\caption{Three candidate devices, and why the choice is not cosmetic. \textbf{(A)} A line junction: the quantum-Hall layer is laterally adjacent and separated from the superconductor by a normal or insulating spacer, so its charged modes survive and $G = |\nu_{\mathrm{IQH}}|e^2/h$. This is the geometry adopted, and the object treated by the line-junction kinetic theory of Refs.~\cite{Protopopov2017,Nosiglia2018}. \textbf{(B)} Everything on one TI surface, where central charges attach to domain walls rather than regions: SC$\,|\,$magnetic walls carry $2n+1$ chiral Majoranas, $c = n+\half$, and magnetic$\,|\,$magnetic walls carry $|\Delta n|$ chiral Dirac modes. \textbf{(C)} Direct contact, in which the quantum-Hall edge is Andreev-converted into chiral Andreev edge states. The conversion takes $|\nu|$ Dirac modes to $2|\nu|$ Majorana modes and therefore leaves the net chiral central charge, and $\kappa_{\mathrm{th}}$, exactly unchanged --- an instance of Corollary~\ref{cor:impossibility} --- while destroying the charge sector on which the diagnostics of Sec.~\ref{sec:diagnostic} depend.}
\label{fig:geometries}
\end{figure*}

\emph{The equilibrated value in geometry (A).} In geometry (A) the TI--SC boundary and the quantum-Hall edge co-propagate along a junction of length $L$ exceeding both the charge and thermal equilibration lengths $\ell_{\mathrm{eq}}^{\mathrm{ch}}, \ell_{\mathrm{eq}}^{\mathrm{th}}$, and the setup is the direct analog of the $(\nu,\nu')$ junction of Roy \emph{et al.}, with the role of one of the integer-filled regions played by the TI--SC heterostructure. The quantum-Hall region carries net central charge $|c_{\mathrm{IQH}}| = |\nu_{\mathrm{IQH}}|$. Under full thermal equilibration, the equilibrated two-terminal value follows from the series function \eqref{eq:R3arm} with the TI--SC boundary forming the two \emph{flanking} segments ($c_{\parallel} = c_M = \half$) and the IQH region the \emph{central} one ($c_{\perp} = |\nu_{\mathrm{IQH}}|$), i.e.\ a TI--SC\,$|$\,IQH\,$|$\,TI--SC chain:
\begin{equation}
\kappa_{\mathrm{th}}^{\mathrm{kin}}\big/\kzero T \;=\; \frac{c_M \,|\nu_{\mathrm{IQH}}|}{c_M + 2|\nu_{\mathrm{IQH}}|} \;=\; \frac{|\nu_{\mathrm{IQH}}|}{1 + 4|\nu_{\mathrm{IQH}}|}.
\label{eq:kappa_kin}
\end{equation}
For $|\nu_{\mathrm{IQH}}| = 1$ this gives $\kappa_{\mathrm{th}}^{\mathrm{kin}} = \tfrac{1}{5}\kzero T$; for $|\nu_{\mathrm{IQH}}| = 2$, $\tfrac{2}{9}\kzero T$; for $|\nu_{\mathrm{IQH}}| = 3$, $\tfrac{3}{13}\kzero T$. These are robust plateau values, requiring no fine-tuning, and they arise from boundary kinetics rather than bulk topology, exactly in the spirit of Roy \emph{et al.}

The arrangement is a physical statement about the device, not a free choice, and by the asymmetry of \eqref{eq:R3arm} it matters: the converse chain IQH\,$|$\,TI--SC\,$|$\,IQH would give $\kappa_{\mathrm{th}}^{\mathrm{kin}}/\kzero T = |\nu_{\mathrm{IQH}}|/\bigl(2(|\nu_{\mathrm{IQH}}|+1)\bigr)$, i.e.\ $\tfrac{1}{4}, \tfrac{1}{3}, \tfrac{3}{8}$ for $|\nu_{\mathrm{IQH}}| = 1,2,3$. Which of the two obtains is fixed by the device layout (which phase is the central region) together with how the Majorana and IQH branches thermalize at the contacts, the same question that enters as the single-temperature hypothesis of Theorem~\ref{thm:pinning}. We adopt the flanking arrangement \eqref{eq:kappa_kin} as the natural one for a quantum-Hall strip running between two TI--SC regions, each separated from it by its own spacer, and defer its first-principles justification, together with the replacement of this heat-balance heuristic by a kinetic-equation treatment of the composite boundary, to Sec.~\ref{sec:open}.

We restate, in the interest of leaving no ambiguity, the assumptions on which the device-specific prediction \eqref{eq:kappa_kin} rests, since they are stronger than those of the universal theorem. (a) \emph{Mode content}: geometry (A), with a spacer for which some width $w$ simultaneously prevents Andreev conversion of the quantum-Hall edge and permits inter-edge equilibration over the junction length. We do not establish that such a window exists for realistic parameters; that is the first open problem of Sec.~\ref{sec:open}, and everything below is conditional on it. (b) Full \emph{thermal} equilibration along the junction, $L \gg \ell_{\mathrm{eq}}^{\mathrm{th}}$. (c) An ideal floating contact that thermalizes the composite Majorana-plus-IQH segment to a \emph{single} temperature $T_M$ --- the hypothesis of Theorem~\ref{thm:pinning}, here applied to a segment of mixed charged/neutral character for which it is not guaranteed. (d) The arrangement $(c_\parallel, c_\perp) = (c_M, |\nu_{\mathrm{IQH}}|)$, TI--SC flanking, rather than its converse. Under (a)--(d), Eq.~\eqref{eq:kappa_kin} is an instance of the pinned value $\mathcal{R}(|c_i|)\,\kzero T$ of Theorem~\ref{thm:pinning}, with $\mathcal{R}$ the three-segment function \eqref{eq:R3arm}. The universal theorem of Sec.~\ref{sec:main} requires none of (a)--(d) in device-specific form; only their generic counterparts. Assumption (a) is of a different kind from the rest: (b)--(d) are questions about how a stated device behaves, while (a) is the question of what the device is.

\begin{figure}[t]
\centering
\includegraphics[width=0.98\columnwidth]{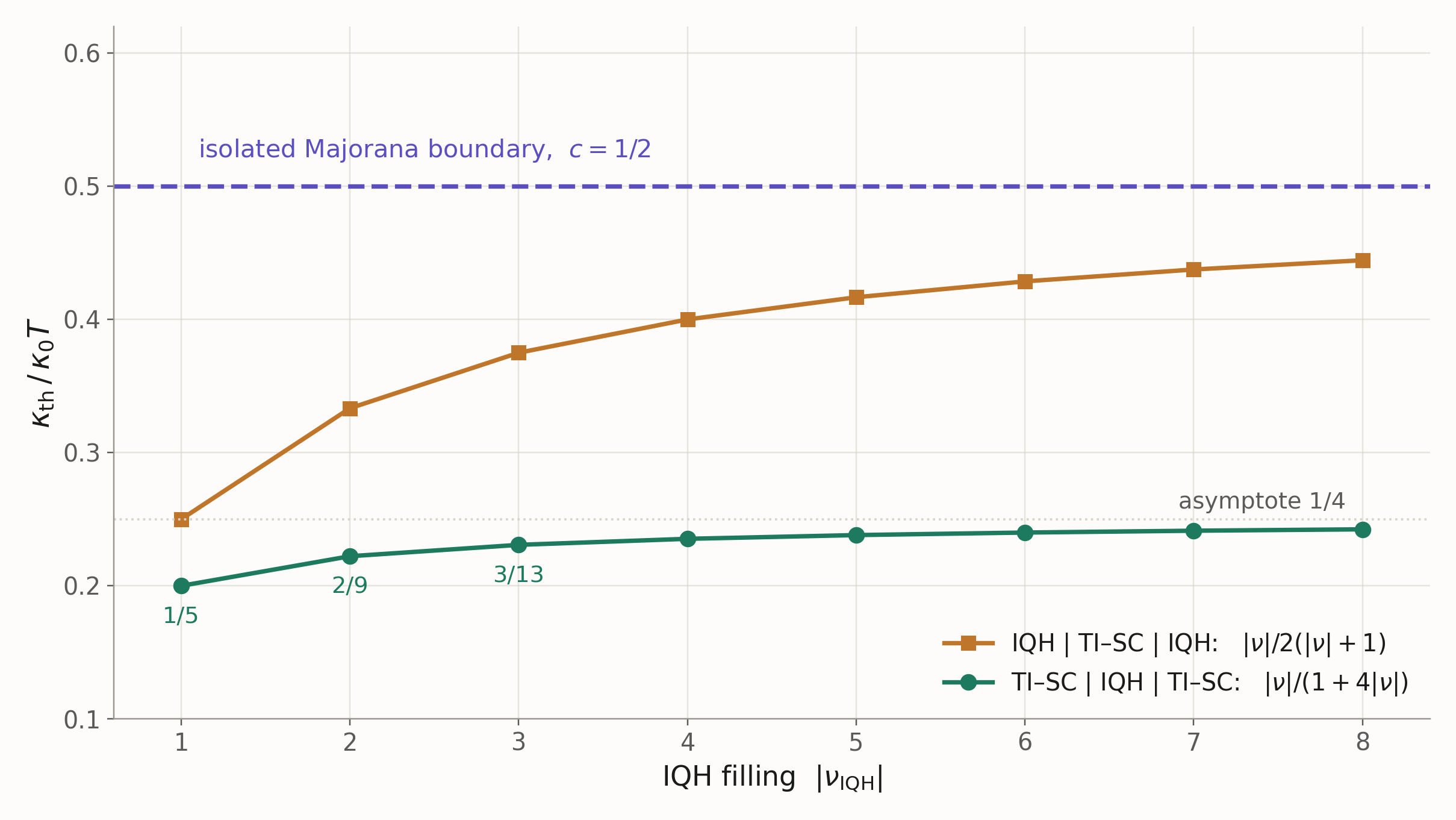}
\caption{Equilibrated plateau values in geometry (A), for the two arrangements of the three-segment chain, against the isolated-boundary value $\half\kzero T$. The flanking family saturates at $\tfrac14$; the converse rises toward $\half$ and approaches the isolated-boundary value asymptotically without reaching it, a near-degeneracy at large filling --- the separation is $1/(|\nu_{\mathrm{IQH}}|+1)$, about $11\%$ at $|\nu_{\mathrm{IQH}}| = 8$ --- of exactly the kind Corollary~\ref{cor:impossibility} concerns. In this platform $\kappa_{\mathrm{th}}$ does discriminate between the isolated and equilibrated regimes at accessible fillings. What Corollary~\ref{cor:impossibility} forbids is finer --- two realizations presenting the same anomaly data to the same network.}
\label{fig:values}
\end{figure}

\emph{What the theorem does and does not say.} With both realizations in hand, we state the exact scope of Corollary~\ref{cor:impossibility}, since it is easy to overstate. The topological value \eqref{eq:kappa_top} is $\half\kzero T$ and belongs to the \emph{isolated} chiral boundary; the kinetic values \eqref{eq:kappa_kin} are $\tfrac15, \tfrac29, \tfrac{3}{13}$ and belong to the \emph{equilibrated} TI--SC\,$|$\,IQH\,$|$\,TI--SC chain. These are different numbers, and they are different because they are different networks. In this platform, therefore, $\kappa_{\mathrm{th}}$ does discriminate between the two regimes: an experiment returning $\tfrac15$ rather than $\tfrac12$ has identified the equilibrated one, and we do not claim otherwise.

We put a number on that statement, because it is not uniform in filling. Measured as a fraction of the isolated-boundary value, the separation between the two regimes is
\begin{equation}
\frac{\half - \mathcal{R}}{\half} \;=\;
\begin{cases}
\dfrac{1+2|\nu_{\mathrm{IQH}}|}{1+4|\nu_{\mathrm{IQH}}|}, & \text{flanking, Eq.~\eqref{eq:kappa_kin}},\\[2.2ex]
\dfrac{1}{|\nu_{\mathrm{IQH}}|+1}, & \text{converse arrangement}.
\end{cases}
\label{eq:separation}
\end{equation}
The flanking family is separated from $\half\kzero T$ by more than $50\%$ at every filling and is distinguishable outright. The converse family is not: its separation falls as $1/(|\nu_{\mathrm{IQH}}|+1)$, reaching $11\%$ at $|\nu_{\mathrm{IQH}}| = 8$ and dropping below the ${\sim}10\%$ accuracy demonstrated in Ref.~\cite{RoyNatComm2026} for $|\nu_{\mathrm{IQH}}| \gtrsim 9$. Beyond that filling the equilibrated junction and the isolated chiral Majorana boundary are experimentally degenerate in the thermal sector, and the charge-sector diagnostics of Sec.~\ref{sec:diagnostic} become necessary in practice even though Corollary~\ref{cor:impossibility} does not make them necessary in principle. The discrimination claimed above is thus a small-filling statement for the converse arrangement and an all-filling statement for the flanking one.

What Corollary~\ref{cor:impossibility} forbids is finer and, for the field, more consequential: two realizations that present the \emph{same} anomaly data to the \emph{same} network are indistinguishable by $\kappa_{\mathrm{th}}$. That is exactly the situation in the graphene device of Roy \emph{et al.}, where $\mathcal{R}(2,1) = \half$ numerically equals $c_M$, and it will be the situation in any future platform where the two happen to coincide. The theorem does not explain that coincidence --- $\mathcal{R}(2,1) = \half$ is an arithmetic fact about one filling combination --- but it shows that no thermometry can undo it, and that the search for a mechanism-sensitive observable must therefore leave the thermal sector altogether. An earlier version of this analysis treated the two mechanisms as \emph{additive} contributions to a total $\kappa_{\mathrm{th}}$ and proposed a magnetic-field scan of $\kappa_{\mathrm{th}}$ as a separator; the theorem shows both moves were wrong in a productive way. In a given equilibrated geometry there is \emph{one} thermal conductance, fixed by anomaly data and circuit topology; ``topological'' and ``kinetic'' label \emph{regimes} (the isolated boundary versus the equilibrated junction), not additive channels; and any $B$-dependence of $\kappa_{\mathrm{th}}$ traces steps in the anomaly data themselves (IQH plateau transitions, or $\mathcal{C}_M$ transitions in \eqref{eq:c_total}), never the mechanism. The mechanism lives in the charge sector, to which we now turn.

\section{Diagnostic signatures: the unique charge-sector resolution}
\label{sec:diagnostic}

By Remark~\ref{rem:notpinned}, the topology--versus--kinetics distinction cannot be read from $\kappa_{\mathrm{th}}$ and must be sought in observables that are not anomaly-pinned. At a superconducting boundary there are exactly two such observables accessible in a two-terminal transport measurement: the electrical conductance (through the Lorenz ratio) and the current noise (through the Fano factor). We treat them in turn and argue that together they are not merely \emph{sufficient} but \emph{necessary} --- any complete resolution must use both, because each is blind to a case the other resolves.

\subsection{Lorenz ratio}

The two mechanisms have qualitatively different Wiedemann--Franz status. The topological response \eqref{eq:kappa_top} arises from a boundary that carries heat but no conserved charge: the vortex carries $e/4$ as an \emph{integrated} polarization charge, localized rather than transported, and contributes no corresponding fractional electrical conductance to the two-terminal measurement. Consequently the Lorenz ratio of the topological response is anomalous,
\begin{equation}
L^{\mathrm{top}} \;\equiv\; \kappa_{\mathrm{th}}^{\mathrm{top}}\big/(T G^{\mathrm{top}}) \;\to\; \infty
\label{eq:L_top}
\end{equation}
in the sense that $G^{\mathrm{top}} \to 0$ while $\kappa_{\mathrm{th}}^{\mathrm{top}}$ remains finite. This is in line with the broader observation that for topological superconductors lacking conserved charge, thermal transport is the only probe of topology \cite{ryu2012electromagnetic,SekineNomura2021}.

The kinetic contribution \eqref{eq:kappa_kin}, by contrast, is carried by equilibrated electron-like modes and satisfies the Wiedemann--Franz law,
\begin{equation}
L^{\mathrm{kin}} \;=\; L_0 \;=\; \pi^2 \kB^2/3 e^2,
\label{eq:L_kin}
\end{equation}
exactly in parallel with Roy \emph{et al.}\ \cite{RoyNatComm2026}. This is the central distinguishing signature: measuring the electrical conductance of the TI--SC--IQH boundary in parallel with the thermal conductance separates the two mechanisms directly. If the measured Lorenz ratio of the fractional response is $L_0$, the response is kinetic; if it is anomalous, the response is topological; and at partial equilibration the measured $L$ interpolates as a function of junction length $L$ (Remark~\ref{rem:assumptions}), with the crossover profile itself a quantitative prediction of the (open) kinetic theory.

Two qualifications are in order. First, \eqref{eq:L_top} refers to the isolated TI--SC boundary, where the chiral Majorana mode is the only channel. In the composite TI--SC\,$|$\,IQH\,$|$\,TI--SC device the two-terminal electrical conductance is that of the IQH region, $G = |\nu_{\mathrm{IQH}}|e^2/h$, so that with \eqref{eq:kappa_kin} the measured ratio is neither limit but the definite fraction
\begin{equation}
L \;=\; \frac{\kappa_{\mathrm{th}}}{T G} \;=\; \frac{L_0}{1 + 4|\nu_{\mathrm{IQH}}|},
\label{eq:L_device}
\end{equation}
i.e.\ $L_0/5$, $L_0/9$, $L_0/13$ for $|\nu_{\mathrm{IQH}}| = 1,2,3$; in geometry (B) the same computation gives $L = m L_0/(m+4|\Delta n|)$. Both expressions presuppose that the charged modes reach the contacts as charged modes, i.e.\ geometry (A) or (B) and not (C): in (C) the numerator of \eqref{eq:L_device} is unchanged and the denominator is not, so a kinetic device would present an anomalous Lorenz ratio and imitate the topological one. This is the sharpest practical reason for the spacer. Equation \eqref{eq:L_device} is the quantity an experiment on the composite device actually returns, and its departure from $L_0$ is a consequence of the series geometry rather than of the mechanism; the mechanism-sensitive statement is the comparison of the measured $G$ with the anomaly-pinned $\kappa_{\mathrm{th}}$, not the numerical value of $L$ by itself. Second, and for the same reason, since $\kappa_{\mathrm{th}}$ is anomaly-pinned the Lorenz ratio carries no information beyond the electrical conductance: measuring $L$ is measuring $G$ with a known multiplier. We keep the Lorenz language because that is how the comparison is made experimentally, but its content lies wholly in the charge sector --- which is the point of Remark~\ref{rem:notpinned}, and the reason the shot noise of Sec.~\ref{sec:diagnostic}.B is not redundant with it.

\subsection{Shot noise}

A second, sharper diagnostic is the excess shot noise at the detector contacts, computed within the framework of Manna, Das, Goldstein, and Gefen \cite{MannaDasGoldsteinGefen2024a,MannaDasGoldsteinGefen2024b,MannaDasGoldsteinGefen2025}. For the kinetic mechanism, the relevant Fano factor is that of equilibrated normal carriers, and the relationship between excess noise $S$ at drains $D_1, D_2$ and the floating-contact temperature $T_M$ takes the form (cf.\ Eq.~(12) of \cite{RoyNatComm2026})
\begin{equation}
S^{\mathrm{kin}}_{\mathrm{Linear}} \;\equiv\; S_{D_1} - \tfrac{1}{4} S_{D_2} \;=\; |\nu_{\mathrm{IQH}}|\, \frac{e^2}{h}\, \kB (T_M - T_0),
\label{eq:S_kin}
\end{equation}
linear in $T_M$ and proportional to the IQH conductance. The particular drain combination, and the weight $\tfrac14$ in it, are those of Eq.~(12) of Ref.~\cite{RoyNatComm2026}: they are properties of that measurement geometry, chosen so that the partitioning contribution of the injected current cancels and the thermal signal survives, and they carry no mechanism dependence. For the topological mechanism, the equivalent signal arises from quasiparticle creation at the vortex defects, with charge $e/4$:
\begin{equation}
S^{\mathrm{top}}_{\mathrm{Linear}} \;=\; \tfrac{1}{4}\,N_v^{1/2}\, \frac{e^2}{h}\, \kB (T_M - T_0) \cdot \mathcal{F}_v,
\label{eq:S_top}
\end{equation}
where $N_v = BA/\Phi_0$ is the number of flux quanta threading the interface area $A$ enclosed by the floating contact, and $\mathcal{F}_v = O(1)$ is the Fano factor of $e/4$ vortex tunneling at the interface, to be computed in the dilute-vortex limit by the methods of \cite{ZhangGornyiGefen2025}. The factor $\tfrac{1}{4}$ is the vortex charge in units of $e$, and the $N_v^{1/2}$ scaling reflects independent contributions from each vortex. That independence deserves scrutiny, and we mark it as provisional. The vortices do not form an uncorrelated gas: they sit on an Abrikosov lattice whose positional order is maintained by the London repulsion, and, carrying the charge \eqref{eq:e_over_4}, they repel electrostatically as well --- a repulsion screened by the Meissner effect and therefore of the same range $\lambda$ as the magnetic one, rather than long-ranged Coulomb. Whether tunneling events at distinct vortices are statistically independent under these correlations, and hence whether the exponent in \eqref{eq:S_top} is $\tfrac12$ rather than the value $1$ that strictly independent noise sources would give, is not settled here; it is part of the $\mathcal{F}_v$ computation. $\mathcal{F}_v$ is the sole material-dependent parameter, and both it and the scaling exponent are open tasks (Sec.~\ref{sec:open}). Equation \eqref{eq:S_top} is the direct analog, in this platform, of the noise diagnostic developed for fractional charges in FQH systems.

The key experimental upshot is that the two signals \eqref{eq:S_kin} and \eqref{eq:S_top} have different scaling with $B$: the kinetic signal is piecewise constant in $B$ between IQH plateaus, while the topological signal grows as $B^{1/2}$. Note that this $B$-scan diagnostic lives entirely in the \emph{noise}, consistently with Corollary~\ref{cor:impossibility}: the thermal conductance itself is $B$-independent within a plateau regime. The ratio $S/(T_M-T_0)$ as a function of $B$, measured along with the conductance, fully separates the contributions.

\begin{figure}[t]
\centering
\includegraphics[width=0.98\columnwidth]{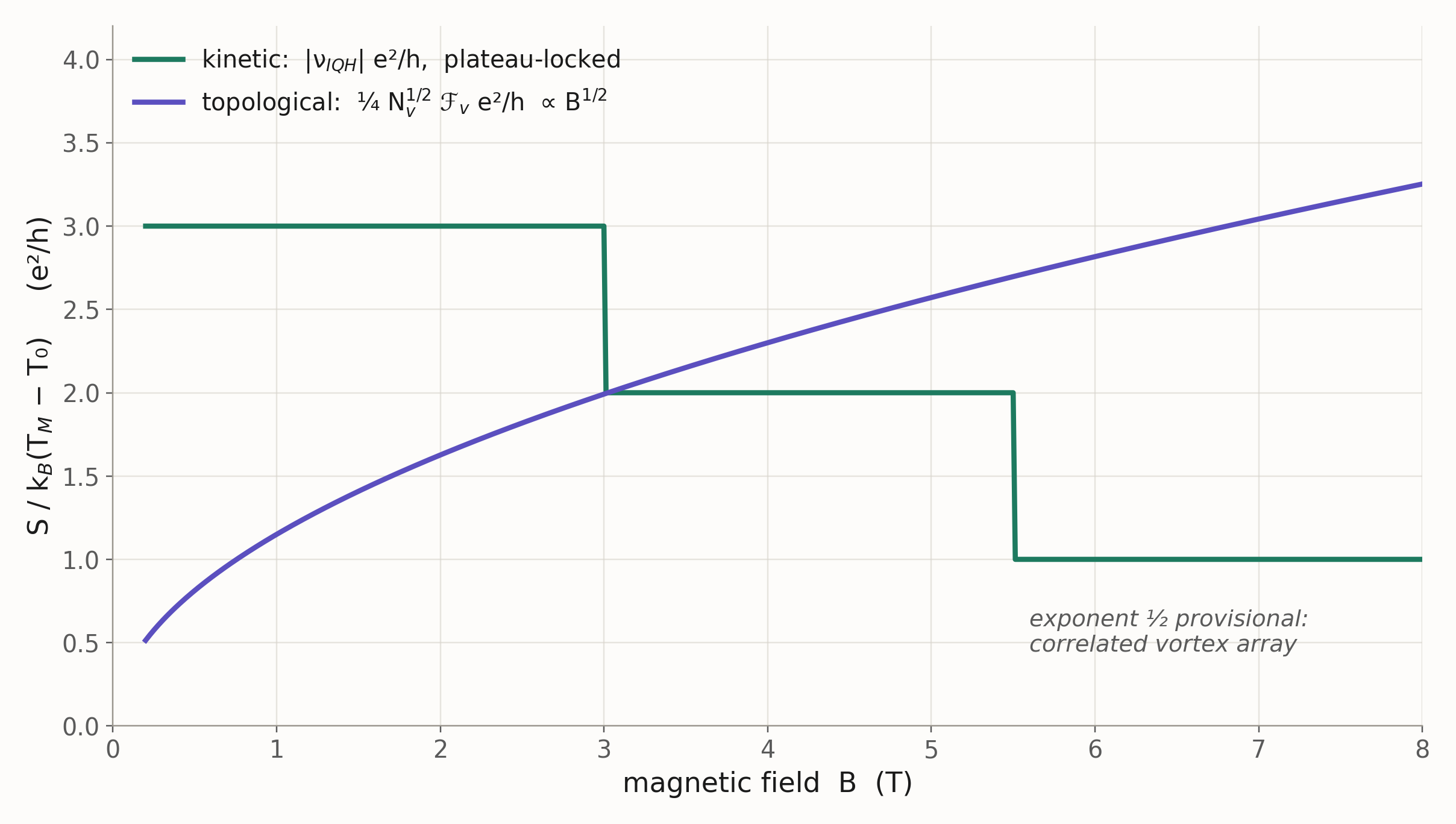}
\caption{The $B$-scan diagnostic, schematically. The kinetic signal \eqref{eq:S_kin} is piecewise constant between IQH plateaus; the topological signal \eqref{eq:S_top} grows with $B$ through the vortex number $N_v = BA/\Phi_0$. The exponent shown is the independent-tunneler value $\tfrac12$ and is provisional: the vortices form a correlated array, and settling the exponent is part of the $\mathcal{F}_v$ computation (Sec.~\ref{sec:open}). Note that the diagnostic lives entirely in the noise --- $\kappa_{\mathrm{th}}$ itself is $B$-independent within a plateau regime, as Corollary~\ref{cor:impossibility} requires.}
\label{fig:noise}
\end{figure}

\subsection{Connection to charge and neutral modes}

Roy \emph{et al.}\ remark in their Discussion that generalizations to FQH boundaries with charge and neutral modes \cite{WangMeirGefen2013,MooreHaldane1997,feldman2021fractional} should follow from their framework. The TI--SC interface is a natural setting for this generalization: the vortex sector contributes an effective neutral mode (the Majorana zero modes and their hybridized bands), while the $e/4$ polarization charge plays the role of a fractionally charged mode. Theorem~\ref{thm:pinning} can be read as the sharp statement of the Roy \emph{et al.}\ program in this setting: the thermal sector is universal and anomaly-pinned across all such realizations, and the charge/neutral decomposition is visible only in the Lorenz ratio and the noise.

\subsection{On measurability}

The diagnostics of this section are charge-sector observables, and we stress that they are measurable \emph{in principle}, so that Corollary~\ref{cor:impossibility} is not vacuous: it says the mechanism is invisible to $\kappa_{\mathrm{th}}$ but visible in quantities that carry independent physical content. Both diagnostics have precedents on platforms of the general type invoked here --- Johnson--Nyquist noise thermometry at a floating contact in the three-arm geometry of Roy \emph{et al.}\ \cite{RoyNatComm2026}, itself built on the graphene quantized-thermal-conductance method of Srivastav \emph{et al.}\ \cite{Srivastav2019}, and the cross-correlation noise measurements of the anyon-collider line \cite{Bartolomei2020} --- and the ingredients (a proximitized three-dimensional TI, an adjacent quantum-Hall region, a floating contact) have each been realized separately. We make no claim about the feasibility of the specific composite device, which lies outside our competence and, we suspect, at or beyond the current experimental frontier; the content of this paper is theoretical, and its experimental import is the conditional statement that \emph{if} such observables are accessed, they --- and only they --- resolve the topology--versus--kinetics question that $\kappa_{\mathrm{th}}$ provably cannot.

\section{Discussion}
\label{sec:discussion}

We have shown that the topological insulator--superconductor interface provides a clean platform for posing the topology--versus--kinetics question raised by Roy \emph{et al.}, and that the question cannot be answered by thermal conductance at all. Theorem~\ref{thm:pinning} shows why: under full equilibration $\kappa_{\mathrm{th}}$ sees only anomaly data and connectivity, so wherever two mechanisms present the same anomaly data to the same network, as they do in the graphene device, the coincidence is beyond the reach of thermometry. It is not the coincidence that is forced, but its irreparability. The resolution is pushed into the charge sector, where the electrical conductance and the shot-noise $B$-scaling separate the mechanisms.

Roy \emph{et al.}\ point at the Bi-bilayer SC of G{\l}odzik and Sedlmayr \cite{GlodzikSedlmayr2023} as another platform where half-integer $\kappa_{\mathrm{th}}$ may arise without Majorana modes. The general lesson, which Theorem~\ref{thm:pinning} makes precise, is that the question \emph{``is this fractional response topological or kinetic?''} \emph{cannot} be answered by $\kappa_{\mathrm{th}}$, however accurately measured, in any equilibrated geometry: the answer lives in the Lorenz ratio and the shot-noise Fano factor. We expect this distinction to play an increasingly central role as more platforms exhibiting fractional thermal response are identified.

\section{Open problems and future directions}
\label{sec:open}

Three problems remain open, and their resolution would complete the theory. We state them precisely, because each is self-contained and centrally relevant, rather than a matter of refinement, and we order them by priority: the first must be settled before the second is even well posed.

\emph{The device and its mode content.} Everything in Sec.~\ref{sec:kinetic} presupposes a composite segment consisting of a charge-blind Majorana branch and a charge-carrying branch of $|\nu_{\mathrm{IQH}}|$ modes, co-propagating and coupled. Whether that object exists is a question about the device, and it has a nontrivial answer. In geometry (C) it does not: a quantum-Hall edge abutting a superconductor is Andreev-converted into chiral Andreev edge states, and the charged branch is simply not there. In geometry (A) it does, provided the spacer width $w$ admits a window --- wide enough to prevent conversion, narrow enough that tunneling and Coulomb coupling equilibrate the two edges over the junction length $L$. Establishing that window, for a stated material stack, is the first task, and it is a microscopic calculation rather than a matter of framing: the conversion length and the inter-edge equilibration length must be estimated in the same model and compared. Geometry (B) poses the corresponding question for magnetic domain walls on a single TI surface, where the charges are $c = n+\half$ and $c = |\Delta n|$ and the plateau values are \eqref{eq:kappa_surface}.

We stress what is and is not at risk here. Theorem~\ref{thm:pinning}, Lemma~\ref{lem:contact} and Corollary~\ref{cor:impossibility} are untouched, as are the $e/4$ charge and the boundary value \eqref{eq:kappa_top}; what depends on the answer is the assignment of central charges to segments, hence \eqref{eq:kappa_kin}, the assumption list (a)--(d), and the reference value in \eqref{eq:L_device}. And the theorem already tells us the shape of the worst case: because Andreev conversion preserves the net chiral central charge, a device that fails this test does not announce itself in $\kappa_{\mathrm{th}}$ --- it announces itself only in the charge sector, which is precisely what the failure destroys. This is why we place the problem first rather than treating it as a caveat.

\emph{The composite-contact kinetic theory.} Theorem~\ref{thm:pinning} applies to the TI--SC--IQH device only under its single-contact-temperature hypothesis: that the charge-blind Majorana branch and the charge-carrying IQH branch, meeting at a common floating contact, thermalize to one temperature. Proposition~\ref{prop:composite} makes the stakes precise. If the branches do thermalize [case (i)], the contact presents a single temperature to the network, Lemma~\ref{lem:contact} applies verbatim, and the plateau values \eqref{eq:kappa_kin} follow. If they do not [case (ii)], the hypothesis fails at the contact, anomaly pinning is broken there, and $\kappa_{\mathrm{th}}$ acquires connectivity dependence beyond the anomaly data --- a value that is not even determined by $\{\half,|\nu_{\mathrm{IQH}}|\}$ alone but depends on the return-path model, as exhibited in Proposition~\ref{prop:composite}. Which alternative the physical Andreev/ohmic contact realizes is not settled by the boundary-CFT input of Sec.~\ref{sec:cft}; it requires a kinetic-equation treatment of the composite boundary, replacing the heat-balance heuristic of Appendix~\ref{app:heatbalance} and resolving the contact-thermalization question directly. This problem is posed within geometry (A) and presupposes the first: if the composite segment does not exist with the stated content, the fork is a question about an object that is not there.

Such a treatment would do three things at once. First, it would decide between cases (i) and (ii) of Proposition~\ref{prop:composite}, and thereby determine whether the TI--SC--IQH platform falls under the anomaly-pinning theorem at all or instead realizes a concrete evasion of it through contact physics. Second, it would fix the flanking-versus-central arrangement in \eqref{eq:kappa_kin} from first principles, rather than by the natural-geometry choice we adopt in Sec.~\ref{sec:kinetic}. Third, through the partial-equilibration crossover, it would produce the length-dependent Lorenz ratio: because the Majorana branch carries heat but no charge, the thermal and charge equilibration lengths decouple parametrically, the regime charted for FQH edges in \cite{Protopopov2017,Nosiglia2018,Srivastav2021}, so the measured Lorenz ratio flows with junction length between the two limits \eqref{eq:L_kin} and \eqref{eq:L_top} rather than sitting at either. This decoupling has no analog in the all-charged graphene junction of Roy \emph{et al.}, and is specific to the superconducting platform. The whole is a line-junction kinetic-equation problem of exactly the type developed by the Gefen--Mirlin school, and it is the piece that promotes the equilibrated-limit theorem of this paper to a statement about the physical device. The stakes are bounded either way, as discussed after Proposition~\ref{prop:composite}.

\emph{The vortex Fano factor.} The factor $\mathcal{F}_v$ for $e/4$ vortex tunneling in \eqref{eq:S_top} is the sole material-dependent unknown in the shot-noise diagnostic and has not been computed. The out-of-equilibrium anyonic-sea technology of \cite{ZhangGornyiGefen2025} should apply, with modifications appropriate to the composite nature of the vortex excitation: it is fractionally charged but not a standard FQH anyon, and its braiding statistics involve the Majorana sector in a way that distinguishes it from $\nu=5/2$ Pfaffian quasiparticles. This is a self-contained computation within the Manna--Das--Goldstein--Gefen framework \cite{MannaDasGoldsteinGefen2024a,MannaDasGoldsteinGefen2025} and completes the noise diagnostic. It carries one further requirement, noted in Sec.~\ref{sec:diagnostic}: the calculation must be performed for a \emph{correlated} vortex array --- positionally ordered by the London repulsion and by the Meissner-screened repulsion between the $e/4$ charges --- rather than for independent tunnelers, and it should therefore settle the scaling exponent in \eqref{eq:S_top} along with the Fano factor itself. Since the $B$-dependence of the topological noise is what distinguishes it from the plateau-locked kinetic signal, that exponent is essential to the diagnostic and not a detail.

Three further directions are worth recording but are secondary. A microscopic lattice calculation of $\mathcal{C}_M$ for the proximity-induced Fu--Kane vortex lattice, on the lines of Murray and Vafek \cite{MurrayVafek2015} for chiral $p$-wave SCs, would determine which regime of \eqref{eq:c_total} the physical system occupies; if $\mathcal{C}_M = \pm 1$ rather than $0$, the topological boundary value becomes integer or zero and the headline coincidence dissolves --- a different paper. In the strong-proximity limit (narrow spacer, or geometry (C)) the TI--SC and quantum-Hall edges merge into a single hybrid Majorana boundary; Theorem~\ref{thm:pinning} still applies, since the hybrid edge has a well-defined net central charge, but the segment decomposition becomes delicate and the charge-sector diagnostics are lost. And the framework suggests a generalization to colliding beams of $e/4$ vortex excitations, analogous to Bartolomei \emph{et al.}\ \cite{Bartolomei2020} but with cross-correlations shaped by the Majorana sector --- a substantial program in its own right, and the natural sequel.

\section*{Acknowledgments}
We acknowledge the foundational role of Roy \emph{et al.}\ \cite{RoyNatComm2026} in motivating the present analysis; we also thank George Batrouni and Brond Larson for ongoing collaboration \cite{JacobsBatrouniLarson2026} on the electromagnetic counterpart of this work, and Yuval Gefen and Sourav Manna for many valuable discussions and exchanges of ideas.  
GRZ thanks Andrea Cappelli for useful discussions and sharing of insights. 

\appendix

\section{Heat balance for the three-arm junction}
\label{app:heatbalance}

The explicit three-segment elimination is given in the proof of Theorem~\ref{thm:pinning} [Eqs.~\eqref{eq:Mbalance}--\eqref{eq:R3arm}], including the flanking-versus-central asymmetry of $\mathcal{R}$ that selects the plateau values \eqref{eq:kappa_kin}. We record here only the physical caveat that motivates the open problem of Sec.~\ref{sec:open}: in the TI--SC--IQH case the segment adjoining the floating contact is a \emph{composite} of a charge-carrying IQH branch and a charge-blind Majorana branch, and the single-contact-temperature hypothesis of Theorem~\ref{thm:pinning} (complete thermalization of both branches to one $T_M$) is not guaranteed, since an Andreev-type contact thermalizes the Majorana branch differently from a normal ohmic contact. The kinetic-equation treatment posed in Sec.~\ref{sec:open} is what settles this question, and with it both the applicability of the theorem to the platform and the arrangement assumed in \eqref{eq:kappa_kin}. The caveat is stated here for geometry (A) of Sec.~\ref{sec:kinetic}, in which the composite segment exists; in geometry (C) it does not arise, because Andreev conversion has already removed the charged branch.

\section{Status of the gravitational-\texorpdfstring{$\theta$}{theta} route}
\label{app:stone}

Stone's analysis \cite{StoneGravitational} of the gravitational thermal Hall effect raises a question about the precise relationship between the gravitational-$\theta_g$ term \eqref{eq:axion_grav} and the laboratory-frame thermal current: a uniform thermal gradient corresponds in the gravitational analog to a uniform Newtonian potential, which does \emph{not} generate a transverse energy current, while only gravitational \emph{gradients} (tidal forces) do. The resolution, as Stone and subsequent authors clarify \cite{StoneGravitational,Kapustin2021}, is that the universal energy current \eqref{eq:IE} underlying all results in this paper is a manifestly physical, frame-independent quantity computed directly from the boundary CFT, and the proper laboratory-frame statement is the Cappelli--Huerta--Zemba formula \eqref{eq:CHZ} as written. The bulk gravitational-$\theta_g$ term is best regarded as a useful organizing principle and effective bookkeeping device rather than as the source of a direct calculational route to the laboratory thermal Hall conductance. We have adopted this view throughout, and Theorem~\ref{thm:pinning} inherits its rigor entirely from the boundary formulation.

\bibliographystyle{apsrev4-2}
\bibliography{tisc_v21}

\end{document}